\RequirePackage{scrlfile}
\PreventPackageFromLoading{natbib}

\documentclass[11pt,a4paper]{article} 

\usepackage{jheppub}

\usepackage{amsthm}

\usepackage[theoremfont]{newtxtext}
\usepackage[upint,smallerops]{newtxmath}

\usepackage{mathtools}
\usepackage[shortlabels]{enumitem}

\usepackage{tikz-cd}
\tikzcdset{row sep/normal=3.5em}

\usepackage{tabu}

\newtheorem{theorem}{Theorem}[section]
\newtheorem{proposition}[theorem]{Proposition}
\newtheorem{lemma}[theorem]{Lemma}
\newtheorem{corollary}[theorem]{Corollary}
\theoremstyle{definition}

\newtheorem{remark}[theorem]{Remark}
\numberwithin{equation}{section}

\usepackage{amsrefs}

\DeclareMathOperator{\gh}{gh}
\DeclareMathOperator{\pa}{p}
\DeclareMathOperator{\ad}{ad}
\DeclareMathOperator{\End}{End}
\DeclareMathOperator{\Pf}{Pf}

\newcommand{\p}{\partial}
\newcommand{\half}{\tfrac{1}{2}}
\newcommand{\R}{\mathbb{R}}
\newcommand{\N}{\mathbb{N}}
\newcommand{\Z}{\mathbb{Z}}
\newcommand{\Om}{\Omega}
\newcommand{\om}{\omega}
\renewcommand{\o}{\otimes}

\newcommand{\<}{\langle}
\renewcommand{\>}{\rangle}
\renewcommand{\(}{\bigl(\!\bigl(}
\renewcommand{\)}{\bigr)\!\bigr)}

\renewcommand{\]}{]\!]}

\newcommand{\g}{\mathfrak{g}}

\renewcommand{\AA}{\mathbb{A}}
\DeclareMathOperator{\cs}{cs}

\newcommand{\CO}{\mathcal{O}}
\newcommand{\CA}{\mathcal{A}}
\newcommand{\CD}{\mathcal{D}}
\newcommand{\CF}{\mathcal{F}}
\newcommand{\FF}{\mathbb{F}}

\newcommand{\tint}{{\textstyle\int}}

\DeclareMathOperator{\pr}{pr}
\newcommand{\h}{\mathsf{X}}

\renewcommand{\H}{\mathsf{H}}
\renewcommand{\P}{\mathsf{P}}

\DeclareMathOperator{\GL}{GL}

\newcommand{\Wedge}{\Lambda}

\newcommand{\s}{\mathsf{s}}

\renewcommand{\dbar}{\overline{\partial}}
\DeclareMathOperator{\Tr}{Tr}
\DeclareMathOperator{\Sym}{Sym}
\DeclareMathOperator{\gl}{\mathbf{gl}}

\DeclareMathOperator{\so}{\mathbf{so}}

\newcommand{\E}{\mathsf{E}}
\newcommand{\f}{\mathsf{f}}

\title{Batalin-Vilkovisky formality for Chern-Simons theory}

\author{Ezra Getzler}

\affiliation{Department of Mathematics, Northwestern University,
  Evanston, Illinois, 60657 USA}

\emailAdd{getzler@northwestern.edu}

\keywords{Chern-Simons theory, Batalin-Vilkovisky formalism, classical
  master equation}

\abstract{We prove that the differential graded Lie algebra of
  functionals associated to the Chern-Simons theory of a semisimple
  Lie algebra is homotopy abelian. For a general field theory, we show
  that the variational complex in the Batalin-Vilkovisky formalism is
  a differential graded Lie algebra.}

\notoc

\begin{document}

\maketitle
\flushbottom

\newpage

\section{Introduction}

In the Batalin-Vilkovisky formalism, functionals form a sheaf of
graded Lie superalgebras over spacetime, with differential given by
the inner derivation with respect to the classical action:
\begin{equation*}
  \s \tint f = ( \tint S , \tint f ) .
\end{equation*}
The functionals are graded by the ghost number, and the antibracket
$(\tint f,\tint g)$ has degree $1$. In particular, cohomology classes
of degree $-1$ form a Lie superalgebra, which generate the symmetries
of the theory, cohomology classes of degree $0$ represent
infinitesimal deformations of the action $\tint S$, and cohomology
classes of degree $1$ represent obstructions. The derivation $\s$ is a
differential because $\int S$ satisfies the Batalin-Vilkovisky classical
master equation
\begin{equation*}
  \bigl( \tint S , \tint S \bigr) = 0 .
\end{equation*}

Following Sullivan \cite{Sullivan} and Goldman and Millson \cite{GM},
it is natural to study whether these differential graded (dg) Lie
superalgebras are \textbf{formal}. A formal dg Lie superalgebra $L$ is
a dg Lie superalgebra that is weakly equivalent to its cohomology
$H^*(L)$. In other words, there is a diagram of dg Lie superalgebras
\begin{equation*}
  \begin{tikzcd}[row sep=1.5em]
    & M \ar[dl] \ar[dr] & \\
    H^*(L) & & L
  \end{tikzcd}
\end{equation*}
in which the arrows induce isomorphisms on cohomology. In particular,
if there is a morphism of dg Lie superalgebras $(H^*(L),0) \to (L,d)$
inducing an isomorphism on cohomology, then $L$ is formal. This is the
only situation that will arise in this article.

Formality is met in the deformation theory of compact Calabi-Yau
manifolds: the dg Lie algebra
$\bigl( \Om^{0,*}(M,\Lambda T),\dbar \bigr)$, namely, the Dolbeault
resolution of the Schouten algebra, is formal. By the Tian-Todorov
theorem, more is true: the bracket induced on the cohomology
$H^*(M,\Lambda T)$ by the Schouten bracket vanishes, so this dg Lie algebra
is in fact weakly equivalent to an abelian graded Lie algebra. We say
that $\Om^{0,*}(M,\Lambda T)$ is \textbf{homotopy abelian}.

In this article, we show that a similar property holds for the
Batalin-Vilkovisky cohomology for the Chern-Simons theory associated
to a semisimple Lie algebra $\g$. Chern-Simons theory is a classical
field theory on $\R^3$ associated to a real Lie algebra $\g$ and an
invariant non-degenerate symmetric bilinear form $\<x,y\>$. Let
$A=A_1dt^1+A_2dt^2+A_3dt^3$ be a connection one-form on $\R^3$, with
values in the Lie algebra $\g$. The Chern-Simons theory has action
\begin{equation*}
  \tint \cs_3(A) ,
\end{equation*}
where $\cs_3(A)$ is the Chern-Simons 3-form
\begin{equation*}
  \cs_3(A) = dt^i \, dt^j \, dt^k \bigl( \half \< A_i , \p_j A_k \>
  + \tfrac16 \< A_i , [A_j , A_k] \> \bigr) .
\end{equation*}
The Euler-Lagrange equation is
\begin{equation*}
  \frac{\delta\cs_3(A)}{\delta A_i} = \half \epsilon^{ijk} F_{jk} = 0 ,
\end{equation*}
where $F_{ij}=\p_iA_j-\p_jA_i+[A_i,A_j]$, and the classical solutions
of the theory are precisely the flat connections.

In addition to the gauge field $A$, Batalin and Vilkovisky consider a
fermionic ghost field $c$, also with values in $\g$, of ghost number
$1$, and antifields $A^{+i}$ and $c^+$, of ghost number $-1$ and $-2$
respectively. Let $\epsilon_{ijk}$ and $\epsilon^{ijk}$ be the antisymmetric tensors
with $\epsilon_{123}=1$ and $\epsilon^{123}=1$: these tensors correspond to a choice
of orientation of $\R^3$. The fields $A$ and $c$ and their antifields
$A^+$ and $c^+$ are the components of a superconnection form
\begin{equation*}
  \AA = c + dt^i A_i +  \half \epsilon_{ijk} dt^idt^j A^{+k} + dt^1dt^2dt^3 c^+ .
\end{equation*}
Let $\cs_3(\AA)$ be the Chern-Simons 3-form, the 3-form component of
the differential form
\begin{equation*}
  \half \< \AA , d\AA \> + \tfrac{1}{6} \< \AA ,
  [\AA,\AA] \> .
\end{equation*}
Axelrod and Singer \cite{AS} identified this form as the solution of
the Batalin-Vilkovisky classical master equation for Chern-Simons
theory:
\begin{equation*}
  \half \( \tint \cs_3(\AA) , \tint \cs_3(\AA) \) = 0 .
\end{equation*}

Let $\CA$ be the algebra of differential polynomials in the
coordinates $\{t^1,t^1,t^3\}$ of $\R^3$ and the fields
$\{c,A_i,A^{+i},c^+\}$, or the polynomial algebra generated by the
coordinates on the jet space
\begin{equation*}
  \{ t^i , \p^\alpha c,\p^\alpha A_i,\p^\alpha A^{+i},\p^\alpha c^+ \mid i\in\{1,2,3\} , \alpha \in
  \N^3 \} .
\end{equation*}
This is a module for the Lie algebra of constant coefficient vector
fields spanned by $\{\p_1,\p_2,\p_3\}$, and we may form the
variational complex
\begin{equation*}
  0 \to \R \to \Om^0 \o \CA \xrightarrow{d} \Om^1 \o_\CO \CA
  \xrightarrow{d} \Om^2 \o_\CA \CA \xrightarrow{d} \Om^3 \o_\CA \CA
  \xrightarrow{\int} \CF \to 0 .
\end{equation*}
Here, $\CO$ is the sheaf of functions on $\R^3$, $\Om^k$ is the sheaf
of $k$-forms on $\R^3$ and $d$ is the de Rham differential
\begin{equation*}
  d = \sum_{i=1}^3 dt^i \p_i .
\end{equation*}
The space $\CF$ of functionals is the cokernel of the de Rham
differential $d:\Om^2 \o_\CO \CA \to \Om^3 \o_\CO \CA$, and $\int$ is the
projection from $\Om^3\o_\CO\CA$ to $\CF$. An important result in
variational calculus, due to Vinogradov and Takens, is that the
variational complex is exact (Olver \cite{Olver}*{Theorem~5.56}).

Let $\FF$ be the truncation of this sequence obtained by deleting the
rightmost term $\CF$:
\begin{equation*}
  0 \to \R \to \Om^0 \o_\CO \CA \xrightarrow{d} \Om^1 \o_\CO \CA
  \xrightarrow{d} \Om^2 \o_\CO \CA \xrightarrow{d} \Om^3 \o_\CO \CA \to 0 .
\end{equation*}
Placing $\Om^k\o_\CO\CA^\ell$ in degree $k+\ell-3$, we obtain a resolution
of $\CF$. In Section 2, we show that the Batalin-Vilkovisky
antibracket on $\CF$ lifts to an antibracket $\(-,-\)$ on $\FF$, which
we call the Soloviev antibracket, making $\FF$ into a differential
graded Lie algebra. This generalizes the case of a single independent
variable \cites{Darboux,covariant}. In this general setting, we also
show that a solution $\tint S$ of the Batalin-Vilkovisky classical
master equation lifts to a Maurer-Cartan element of $\FF$, that is, a
solution of the equation
\begin{equation*}
  dS + \half \(S,S\) = 0 .
\end{equation*}
This implies that the differential graded Lie algebra $\CF$, with
differential $\s$ given by the adjoint action of $\tint S$, is
resolved by the differential graded Lie algebra $\FF$ with
differential $d+\s$, where $\s=\(S,-\)$ is given by the adjoint action
of $S$. The theory that we present is more general than is needed for
our discussion of Chern-Simons theory, but we hope that the general
case will find application in future work.

We also show that the differential $d+\s$ is equivalent to the
differential $d+\h_S$ introduced by Gelfand and Dikii, where $\h_S$ is
the Hamiltonian vector field associated to $S$. The differential
$d+\h_S$ has the advantage that it is the total derivative of a double
complex. This means that it is considerably easier to calculate its
cohomology than to calculate directly the cohomology of $d+\s$. On the
other hand, $d+\h_S$ is not a derivation of the Soloviev antibracket,
meaning it has a less transparent relationship to the differential
graded Lie algebra $\CF$.

In Section 4, following Barnich, Brandt and Henneaux \cite{BBH}, we
calculate the cohomology of the differential $d+\h_\Pi$, where $\Pi$ is
the action of the abelian Chern-Simons theory. From this, we obtain a
quasi-ismorphism from the Chevalley-Eilenberg complex
$\tau_{\ge1}C^{*+3}(\g)$ (with vanishing differential) to $\FF$ with
differential $d+\pi$, where $\pi$ is the differential associated to the
abelian Chern-Simons action $\Pi$. This remains a quasi-isomorphism of
differential graded Lie algebras after passage to the non-abelian
Chern-Simons theory. The differential on $\tau_{\ge1}C^*(\g)$ is now the
Chevalley-Eilenberg differential, allowing us to identify the
cohomology of $\CF$ with $\tau_{\ge1}H^*(\g)$. In particular, if $\g$ is
semisimple, Whitehead's lemma implies that the cohomology
$H^*(\CF,\s)$ vanishes in negative degrees, since $H^1(\g)=H^2(\g)=0$.

In summary, we have a diagram of three differential graded Lie
algebras
\begin{equation*}
  \tau_{\ge1}C^{*+3}(\g) \to \FF^* \to \CF^* ,
\end{equation*}
where both arrows are quasi-isomorphisms. When $\g$ is reductive,
there is a quasi-isomorphism
\begin{equation*}
  \tau_{\ge1}H^*(\g) \cong \tau_{\ge1}C^*(\g)^\g \hookrightarrow \tau_{\ge1}C^*(\g) .
\end{equation*}
In the appendix, we show that the bracket induced on
$\tau_{\ge1}C^*(\g)^\g$ by the bracket on $C^*(\g)^\g$ vanishes if
$\g$ is semisimple. This is proved using Chevalley's theorem, that the
algebra $I(\g)$ of invariant polynomials is spanned by trace
polynomials $\Tr_V(\rho(x)^k)$, where $\rho:\g\to\End(V)$ ranges over
finite-dimensional representations of $\g$ and $k$ over positive
integers. This shows that the sheaf of differential graded Lie
algebras $(\CF,\s)$ is homotopy abelian.

The explicit construction of the morphism of complexes
$\tau_{\ge1}C^{*+3}(\g) \to \FF^*$, and the proof that its composition with
$\int:\FF\to\CF$ is a morphism of differential graded Lie algebras, is the
heart of the paper. This may also be viewed as a lift of Proposition
2.5 of Cattaneo and Felder \cite{CF}, which is a statement about
general AKSZ models, from calculus to variational calculus.

To do this, we refine the classical master equation to incorporate
equivariance with respect to the Lie algebra of constant vector fields
on $\R^3$. Just as the Batalin-Vilkovisky classical master equation is
the Maurer-Cartan equation for a differential graded Lie algebra, this
refinement is the Maurer-Cartan equation for a \textbf{curved}
differential graded Lie algebra structure on the space of local
functionals tensored with the algebra of polynomials in variables
$(u_1,u_2,u_3)$ of degree $2$. We call this the \textbf{covariant}
extension of the classical master equation. Consider the differential
polynomials
\begin{equation*}
  G_i = \half \epsilon_{ijk} \< A^{+j} , A^{+k} \> + \< A_i , c^+ \> , \quad
  i\in\{1,2,3\} .
\end{equation*}
Our main result is that the morphism that takes a function $f(c)$ of
the field $c$, that is, an element $f\in C^*(\g)$, to the functional
\begin{equation*}
  \bigl( \tint G_1 , \bigl( \tint G_2 , \bigl( \tint G_3 , \tint f(c)
  \bigr) \bigr) \bigr) \in \CF^{*-3} .
\end{equation*}
is a quasi-isomorphism of differential graded Lie algebras. We are
even able to construct a lift of this morphism to $\FF$, and show that
this is also a morphism of differential graded Lie algebras, although
in the proof that $(\CF,\s)$ is homotopy abelian, we do not really
need that this lifted morphism preserves Lie brackets. However, this
refinement may be of use in future work, especially in the study of
boundary conditions.

\section{The Soloviev bracket and the Batalin-Vilkovisky classical
  master equation}

Let $M$ be a graded supermanifold, with coordinates $x^a$ of ghost
number $\gh(x^a)\in\Z$ and parity $\pa(x^a)\in\Z/2$. The
Batalin-Vilkovisky supermanifold $T^*[-1]M$ has coordinates $x^a$ and
the dual coordinates (known in physics as antifields) $\xi_a$, with
ghost number $\gh(\xi_a)=-1-\gh(x^a)$ and parity
$\pa(\xi_a)=1-\pa(x^a)$. A graded supermanifold carries an action of the
Lie group $\GL(1)$, in other words, for each nonzero constant $a$, a
dilation by $a$. The fixed-point set is a supermanifold in the usual
sense, and by the body of the graded supermanifold, we mean the body
of this fixed-point set. The body of $M$ has as coordinates the degree
$0$ coordinates of even parity of $M$. A sheaf on $M$ is a sheaf on
its body.

Graded supermanifolds generalize graded manifolds, where the ghost
number (degree) determines the parity. The introduction of a bigrading
allows for fermionic physical fields. In Chern-Simons field theory,
such fields are absent, and the parity of a field is the same as the
reduction of its ghost number modulo 2.

Fix a natural number $n$, the number of independent coordinates, and
let $\CO$ be the structure sheaf of $\R^n$. Consider the product
$\R^n\times T^*[-1]M$, whose elements are functions in the coordinates
$\{t^i,x^a,\xi_a\}$. The coordinates $t^i$ are the independent
variables, and the coordinates $x^a$ and $\xi_a$ are referred to
respectively as fields and antifields, or simply as fields.

Introduce jet coordinates $\{x^a_\alpha,\xi_{a,\alpha}\}$, where
$\alpha$ ranges over positive $n$-dimensional multi-indices:
\begin{equation*}
  \alpha \in \{ (\alpha_1,\dots,\alpha_n) \in \N^n \mid \alpha>0 \} .
\end{equation*}
Let $|\alpha|=\alpha_1+\dots+\alpha_n$; the notation $\alpha>0$ indicates that
$|\alpha|>0$.  Let $\CA$ be the sheaf of graded commutative algebras on
$U\times T^*[-1]M$ that are functions of the coordinates
$\{t^i,\p^\alpha x^a,\p^\alpha \xi_a\mid 1\le i\le 3,\alpha\ge0\}$. (Different choices of
function algebras may be made here: in practice, we choose functions
that are (graded) polynomials in the coordinates of positive degree,
power series in the coordinates of negative degree, and polynomial,
real analytic, or infinitely differentiable in the coordinates of zero
degree and the independent variables.) This is the sheaf of functions
on the jet space $J_\infty(\R^n,M)$ of maps from $\R^n$ to $M$; we may
think of the jet coordinates as the Taylor coefficients of a formal
map from $\R^n$ to $M$:
\begin{align*}
  x^a(t) &= \sum_\alpha \frac{t^\alpha x^a_\alpha}{\alpha!} , &
  \xi_a(t) &= \sum_\alpha \frac{t^\alpha \xi_{a,\alpha}}{\alpha!} .
\end{align*}

The space of vertical vector fields on the jet space is a free module
over $\CA$; denote the basis of vertical vector fields given by
partial derivatives with respect to the jet coordinates by
\begin{align*}
  \p_{\alpha,a} &= \frac{\p~}{\p x^a_\alpha} , & \p^a_\alpha &= \frac{\p~}{\p\xi_{a,\alpha}}
  .
\end{align*}

Introduce the vector fields corresponding to infinitesimal translation
in the coordinate $t^i$:
\begin{equation*}
  \p_i = \frac{\p~}{\p t^i} + \sum_\alpha \bigl( x^a_{\alpha+i} \p_{a,\alpha} +
  \xi_{a,\alpha+i} \p^a_\alpha \bigr) . 
\end{equation*}
Here, $\alpha+i$ is an abbreviation for the multi-index
$(\alpha_1,\dots,\alpha_i+1,\dots,\alpha_n)$.  A vector field $\tau$ is evolutionary if
it commutes with the vector fields $\p_i$. Such a vector field is
determined by its action on the coordinates $\{t^i,x^a,\xi_a\}$: it is
given by the explicit formula
\begin{equation*}
  \tau = f^i \p_i + \pr \bigl( g^a \p_a + h_a \p^a \bigr) = f^i \p_i +
  \sum_\alpha \bigl( \p^\alpha g^a \p_{a,\alpha} + \p^\alpha h_a \p^a_\alpha \bigr) ,
\end{equation*}
where $f^i=\tau(t^i)$, $g^a=\tau(x^a)$, $h_a=\tau(\xi_a)$, and the prolongation
is the vector field
\begin{equation*}
  \pr \bigl( g^a \p_a + h_a \p^a \bigr) = 
  \sum_\alpha \bigl( \p^\alpha g^a \p_{a,\alpha} + \p^\alpha h_a \p^a_\alpha \bigr) .
\end{equation*}

The variational derivatives are the operators on $\CA$ given by the
formulas
\begin{align*}
  \delta_a &= \sum_{\alpha\in\N^n} (-\p)^\alpha \p_{\alpha,a} & \delta^a &= \sum_{\alpha\in\N^n} (-\p)^\alpha \p^a_\alpha .
\end{align*}
These operators are not derivations, but rather differential operators
of infinite order.

Let $\FF^{*,*}$ be the bicomplex of sheaves
\begin{equation*}
  \FF^{-p,q} =
  \begin{cases}
    \Wedge^p \R^n \o \CA^q , & 0\le p\le n , \\
    \R\eta , & \text{$p=n+1$ and $q=0$}.
  \end{cases}
\end{equation*}
If $I=\{i_1<\dots<i_p\}\subset\{1,\dots,n\}$, let $\sigma_I=\sigma_{i_1}\dots\sigma_{i_p}$.
If $p\le n$, a typical element of $\FF^{-p,q}$ is a sum
\begin{equation*}
  f = \sum_{I=\{i_1<\dots i_p\}} \sigma_I f_I ,
\end{equation*}
where $\gh(f_I)=q$. The total degree of $f$ is $q-p$. The differential
$d:\FF^{-p,q}\to\FF^{-p+1,q}$ equals
\begin{equation*}
  d \left( \sigma_I f_I \right) = \sum_{\ell=1}^k (-1)^{\ell-1} \, \sigma_{I\setminus i_\ell}
  \p_{i_\ell}f_I
\end{equation*}
if $p\le n$, while the differential from $\FF^{-n-1,q}\to\FF^{-n,q}$ takes
$\eta$ to $\sigma_1\dots\sigma_n$. The element
$\sigma_{i_1}\dots\sigma_{i_k}f_I$ of $\FF$ corresponds to the differential form
\begin{equation*}
  \bigl( \p_{i_1} \lrcorner \dots \p_{i_k} \lrcorner dt^1 \dots dt^d \bigr) f_I .
\end{equation*}
We will tacitly assume this identification of elements of $\FF$ with
differential forms.

We also consider the subspace $\FF_\circ\subset\FF$ made up of functions having
no explicit dependence on the coordinates $\{t^i\}$, and the
subalgebra $\CA_\circ=\CA\cap\FF_\circ\subset\CA$. On this subcomplex, the
evolutionary vector field $\p_i$ restricts to the vector field
\begin{equation*}
  \pr \bigl( x^a_i \p_a + \xi_{a,i} \p^a \bigr) .
\end{equation*}
The differential $d:\FF^{-p,*}\to\FF^{-p+1,*}$ induces a differential on
the subspace $\FF_\circ$.

The cohomology
\begin{equation*}
  \CF = \FF^{0,*} \bigm/ d\FF^{-1,*} \cong \CA \bigm/ \sum_{i=1}^n \p_i\CA
\end{equation*}
of the differential $d$ in $\FF^{0,*}$ is the sheaf of functionals of
the classical field theory. The projection from $\FF^{0,*}\cong\CA$ to
$\CF$ is denoted by $f\mapsto \int f$: in terms of differential forms, it
corresponds to integration over differential forms of top degree. We
also consider the cohomology
\begin{equation*}
  \CF_\circ = \FF_\circ^{0,*} \bigm/ d\FF_\circ^{-1,*} \cong \CA_\circ \bigm/ \sum_{i=1}^n
  \p_i\CA_\circ
\end{equation*}
of the differential $d$ in $\FF_\circ^{0,*}$, consisting of functionals
with no explicit dependence on the independent variables $\{t^i\}$.

The sheaves $\CF$ and $\CF_\circ$ are the main objects of study of the
variational calculus. The sheaves $\FF$ and $\FF_\circ$ may be used to
study $\CF$ and $\CF_\circ$, by the following algebraic Poincar\'e lemma.

\begin{theorem}
  \label{Poincare}
  Suppose that $U\subset\R^n$ and the body of $M$ are star-shaped. Then the
  complexes
  \begin{multline*}
    0 \to \FF^{-n-1,*}(U\times T^*[-1]M) \xrightarrow{d}
    \FF^{-n,*}(U\times T^*[-1]M) \xrightarrow{d} \cdots \\
    \xrightarrow{d} \FF^{-1,*}(U\times T^*[-1]M) \xrightarrow{d}
    \FF^{0,*}(U\times T^*[-1]M) \xrightarrow{\int} \CF(U\times T^*[-1]M) \to 0 .
  \end{multline*}
  and
  \begin{multline*}
    0 \to \FF_\circ^{-n-1,*}(T^*[-1]M) \xrightarrow{d}
    \FF_\circ^{-n,*}(T^*[-1]M) \xrightarrow{d} \cdots \\ \xrightarrow{d}
    \FF_\circ^{-1,*}(T^*[-1]M) \xrightarrow{d} \FF_\circ^{0,*}(T^*[-1]M)
    \xrightarrow{\int} \CF_\circ(T^*[-1]M) \to 0 .
  \end{multline*}
  are exact.
\end{theorem}

We will need some of the details of the proof of this theorem. The
proof in Olver \cite{Olver}*{Theorem~5.56, (5.109), (5.111)} is based
on explicit formulas (apparently due to Anderson) for a homotopy
operator
\begin{equation*}
  \H : \FF_\circ^{-p,q} \to \FF_\circ^{-p-1,q}
\end{equation*}
and an operator
\begin{equation*}
  \P : \CF_\circ \to \FF_\circ^{0,*}
\end{equation*}
such that the operators $\tint \P : \CF_\circ \to \CF_\circ$,
\begin{equation*}
  \H d + \P \tint : \FF^{0,*}_\circ \to \FF^{0,*}_\circ
\end{equation*}
and
\begin{equation*}
  d\H + \H d : \FF_\circ^{-p,*} \to \FF_\circ^{-p,*} , \qquad p>0 ,
\end{equation*}
all equal the identity. The operator $\H$ is given on
$f\in\FF_\circ^{-p,*}$, $0\le p<n$, by the formula
{\small
\begin{equation*}
  \H f = \sum_{i=1}^n \sigma_i \sum_{0<\beta\le\alpha} \frac{\beta_i}{|\beta|+n-p} \binom{\alpha}{\beta}
  \p^{\beta-i} \int_0^1 \Bigl(  x^a (-\p)^{\alpha-\beta} \p_{a,\alpha} f(\lambda u) + \xi_a
  (-\p)^{\alpha-\beta} \p^a_\alpha f(\lambda u) \Bigr) \frac{d\lambda}{\lambda} ,
\end{equation*}}
and on $\FF_\circ^{-n-1,0}$ by
$H\eta=\sigma_1\dots\sigma_n\in\FF^{-n,0}$. The operator $\P$ is given on
$\int f\in\CF_\circ$, by the formula
\begin{equation*}
  \P \tint f = \int_0^1 \Bigl(  x^a \delta_a f(\lambda u) + \xi_a \delta^a f(\lambda u) \Bigr)
  \frac{d\lambda}{\lambda} .
\end{equation*}

Filter $\FF_\circ$ by the subspaces
\begin{equation*}
  F^p\FF_\circ = \sum_{r\ge p} \FF_\circ^{-r,*} .
\end{equation*}
The operator $\H$ takes $F^p\FF_\circ$ to $F^{p+1}\FF_\circ$, and
$F^p\FF_\circ=0$ for $p>n+1$, and satisfies the formula $d\H+\H d=1$ on
$F^1\FF_\circ$.

The proof of Theorem \ref{Poincare} for $\FF$ uses a more complicated
homotopy $\H^*$, which is the sum of a homotopy $\H$ given by the same
expression as for $\FF_\circ$, plus the homotopy for differential forms on
the star-shaped set $U$ that figures in the usual proof of the
Poincar\'e lemma. The corresponding projection $\P:\CF\to\FF^{0,*}$ is
given by the same expression as for $\CF_0\to\FF_\circ^{0,*}$. The homotopy
$\H^*$ again satisfies the formula $d\H^*+\H^*d=1$ on $F^1\FF$.

The sheaf $\CF$ is a graded Lie superalgebra with respect to the
Batalin and Vilkovisky (anti)bracket
\begin{equation*}
  ( \tint f , \tint g ) = \sum_a (-1)^{(\pa(f)+1)\pa(\xi_a)} \tint \bigl(
  \delta_af \,\delta^a g + (-1)^{\pa(f)} \delta^a f \, \delta_a g \bigr) .
\end{equation*}
A classical field theory is a solution $\int S\in\CF$, of ghost number
$0$ and even parity, to the classical master equation (Maurer-Cartan
equation) of Batalin and Vilkovisky:
\begin{equation}
  \label{Master}
  \half ( \tint S , \tint S ) = 0 .
\end{equation}
We refer to $\tint S$ as the action of the theory. The operator
\begin{equation*}
  \s \tint f = ( \tint S , \tint f )
\end{equation*}
makes $\CF$ into a differential graded Lie superalgebra; the
cohomology $H^*(\CF,\s)$ is a graded Lie superalgebra, called the BV
cohomology of the theory. (Some authors refer to this as the BRST
cohomology.) The Batalin-Vilkovisky antibracket preserves the subsheaf
$\CF_\circ\subset\CF$. We restrict attention in this paper to classical field
theories whose action $\int S\in\CF_\circ$ has no explicit dependence on the
independent variables $\{t^i\}$.

The resolution $\FF$ of $\CF$ is a differential graded Lie
superalgebra, with respect to a bracket introduced by Soloviev
\cite{Soloviev}, and adapted to the Batalin-Vilkovisky formalism (at
least, when $n=1$) in \cites{Darboux,covariant}. This (anti)bracket is
given by the explicit formula
\begin{multline*}
  \( \sigma_I f_I , \sigma_J g_J \) = \sum_a (-1)^{(\pa(f_I)+1)(\pa(x_a)+|J|)} \\
  \sigma_I \sigma_J \sum_{\alpha,\beta\in\N^n} \left( \p^\beta\p_{a,\alpha}f_I \, \p^\alpha\p^a_\beta g_J +
    (-1)^{\pa(f_I)} \p^\beta\p^a_\alpha f_I \, \p^\alpha\p_{a,\beta} g_J \right) .
\end{multline*}
The Soloviev (anti)bracket preserves the subsheaf $\FF_\circ\subset\FF$.
\begin{proposition}
  \label{lift}
  Given a classical field theory, represented by a solution
  $\int S\in\CF_\circ$ of the classical master equation, there is a lift
  $S\in\FF_\circ$ of $\int S$ to a solution of the Maurer-Cartan equation in
  $\FF_\circ$:
  \begin{equation}
    \label{master}
    dS + \half \( S , S \) = 0 .
  \end{equation}
\end{proposition}
\begin{proof}
  Since $\half(\tint S_0,\tint S_0)=0\in\CF_\circ$, we see that there exists
  an element $S_1\in\FF_\circ^{-1,1}$ such that
  \begin{equation}
    \label{S1}
    dS_1+\half\(S_0,S_0\)=0 .
  \end{equation}
  We now define $S_k\in\FF_\circ^{-k,k}$ for $k>1$ by the recursive formula
  \begin{equation*}
    S_k = - \half \sum_{j=0}^{k-1} \H \( S_j , S_{k-j-1} \) .
  \end{equation*}
  Equivalently, $S$ is the fixed point of the equation
  \begin{equation*}
    S = S_0 + S_1 - \half \H \Bigl( \(S,S\) - \(S_0,S_0\) \Bigr) .
  \end{equation*}

  Let $T=dS+\half\(S,S\)$. We see that
  \begin{equation*}
    dS = dS_0 + dS_1 - \half \Bigl( \(S,S\) - \(S_0,S_0\) \Bigr) +
    \half \H \Bigl( d\(S,S\) - d\(S_0,S_0\) \Bigr) .
  \end{equation*}
  Since $dS_0=0$, this equals
  \begin{equation*}
    dS = dS_1 + \( S_0 , S_0 \) - \half \(S,S\)  \H \(S,dS\) .
  \end{equation*}
  By \eqref{S1}, we see that
  \begin{equation*}
    T = - \H \(S,dS\) .
  \end{equation*}
  By the Jacobi relation for the Soloviev bracket, $\(S,\(S,S\)\)=0$,
  hence
  \begin{equation*}
    T = - \H \(S,T\) .
  \end{equation*}
  Since $\H:F^p\FF_\circ\to F^{p+1}\FF_\circ$ raises filtration degree in
  $\FF_\circ$, we see, by induction on $p$, that $T$ lies in
  $F^p\FF_\circ$. Since $F^p\FF_\circ$ vanishes if $p>n+1$, we conclude that
  $T=0$, proving the theorem.
\end{proof}

We call $S$ the Lagrangian density of the theory. The lift of $S$ to
$\FF_\circ$ yields insight into the problem of extending field theories to
manifolds with boundary, and possibly corners. For a different
approach to these questions, see Cattaneo et al. \cite{CMR}.

Define an operator
\begin{equation*}
  \s f = \( S , f \)
\end{equation*}
on $\FF$. By the classical master equation \eqref{master}, the
operator $d+\s$ on $\FF$ is a differential,
\begin{equation*}
  [d,\s] + \s^2 = 0 ,
\end{equation*}
and gives a lift of the differential $\s$ on $\CF$ to $\FF$. Equipped
with this differential, $\FF$ becomes a sheaf of differential graded
Lie algebras.
\begin{proposition}
  The morphism of sheaves of differential graded Lie superalgebras
  from $(\FF,d+\s)$ to $(\CF,\s)$ given by $f\mapsto\tint f$ is a
  quasi-isomorphism.
\end{proposition}

The operator $\s$ is a graded derivation of the graded Lie algebra
$\FF$. There is another operator $\tilde{\s}$ defined using $S$ that
is a graded derivation of the product on $\FF$, namely the Hamiltonian
vector field $\h_S$ associated to $S$. The (Batalin-Vilkovisky)
Hamiltonian vector field associated to an element $f\in\FF$ is the
evolutionary vector field given by the formula
\begin{equation*}
  \h_f(g) = \sum_{I,J\in\N^n} \h_{\sigma_I f_I} \bigl( \sigma_J g_J \bigr) ,
\end{equation*}
where
\begin{align*}
  \h_{\sigma_I f_I} \bigl( \sigma_J g_J \bigr) &= \sum_a
  (-1)^{(\pa(f_I)+1)(\pa(x_a)+|J|)} \sigma_I \sigma_J \pr \bigl( \delta_{a,\alpha}f_I \,
  \p^a g_J + (-1)^{\pa(f_I)} \delta^a_\alpha f_I \, \p_a g_J \bigr) \\
  &= \sum_a (-1)^{(\pa(f_I)+1)(\pa(x_a)+|J|)} \\
  &\qquad \sigma_I \sigma_J \sum_{\alpha,\beta\in\N^n} (-1)^{|\alpha|} \bigl( \p^{\alpha+\beta} \p_{a,\alpha}f_I \, \p^a_\beta
  g_J + (-1)^{\pa(f_I)} \p^{\alpha+\beta}\p^a_\alpha f_I \, \p_{a,\beta} g_J \bigr) .
\end{align*}

To $f\in\FF$, we have now associated two different operators on $\FF$,
namely the infinite-order differential operator $\ad(f)$ and the
vector field $\h_f$. It turns out that these operators are homotopic.
\begin{lemma}
  \label{homotopy}
  Let
  \begin{multline*}
    C \bigl( \sigma_I f_I , \sigma_J g_J \bigr) = \sum_{i=1}^n \sum_a
    (-1)^{(\pa(f_I)+1)(\pa(x_a)+|J|)} \sigma_i \sigma_I \sigma_J \\
    \sum_{\alpha,\beta\in\N^n} \sum_{0<\gamma\le\alpha} \frac{\gamma_i}{|\gamma|} \binom{\alpha}{\gamma}
    (-\p)^{\gamma-i} \bigl( \p^\beta \p_{a,\alpha} f_I \, \p^{\alpha-\gamma} \p^a_\beta g_J +
    (-1)^{\pa(f_I)} \p^\beta \p^a_\alpha f_I \, \p^{\alpha-\gamma} \p_{a,\beta} g_J \bigr) ,
  \end{multline*}
  where $\gamma-i$ is the multi-index $(\gamma_1,\dots,\gamma_i-1,\dots,\gamma_n)$. Then
  \begin{equation*}
    dC(f,g) = \h_fg - \( f , g \) + C(df,g) + (-1)^{\pa(f)} C(f,dg)
  \end{equation*}
\end{lemma}
\begin{proof}
  This is a consequence of the formula
  \begin{align*}
    & \bigl( \h_{\sigma_If_I} - \ad(\sigma_If_I) \bigr) \bigl( \sigma_J g_J \bigr) =
    \sum_I \sum_a (-1)^{(\pa(f_I)+1)(\pa(x_a)+|J|)} \\
    & \qquad \sigma_I \sigma_J \sum_{\alpha,\beta\in\N^n} \sum_{0<\gamma\le\alpha} \textstyle \binom{\alpha}{\gamma} (-\p)^\gamma
    \bigl( \p^\beta \p_{a,\alpha} S_I \, \p^{\alpha-\gamma} \p^a_\beta g_J + (-1)^{\pa(f_I)}
    \p^\beta \p^a_\alpha S_I \, \p^{\alpha-\gamma} \p_{a,\beta} g_J \bigr) . \qedhere
  \end{align*}
\end{proof}
 
Consider the operator
\begin{equation*}
  \tilde{\s} = \h_S = \sum_I \h_{\sigma_IS_I} .
\end{equation*}
The following theorem, relating the action of the operators $\s$ and
$\tilde{\s}$, is the first main result of this paper.
\begin{theorem}
  \label{main}
  Suppose that $U\subset\R^n$ and the body of $M$ is star-shaped. Then there
  is an automorphism
  \begin{equation*}
    \alpha = 1 + \sum_{k=1}^\infty \alpha_k ,
  \end{equation*}
  where
  $\alpha_k:\FF^{p,q}(U\times T^*[-1]M)\to\FF^{p-k,q+k}(U\times T^*[-1]M)$ is an
  operator of even parity, that intertwines the operators $d+\s$ and
  $d+\tilde\s$:
  \begin{equation*}
    (d+\tilde{\s})\alpha = \alpha(d+\s) .
  \end{equation*}
\end{theorem}

Note that $\alpha$ is actually a finite sum, since the operators
$\alpha_k$ vanish for $k>n+1$.

The following corollary also follows from results contained in
\cite{Olver}.
\begin{corollary}
  The operator $d+\tilde{\s}$ is a differential of the sheaf of graded
  superspaces $\FF$, and there is a natural isomorphism
  $H^*(\FF,d+\s)\cong H^*(\FF,d+\tilde{\s})$. 
\end{corollary}

Since $\tilde{\s}$ is evolutionary, we have $[d,\tilde{\s}]=0$, and
hence the above corollary actually establishes that
$\tilde{\s}^2=0$. Another proof of this result is given in
\cite{covariant} for $n=1$.

\begin{proof}[Proof of Theorem \ref{main}]
  Write
  \begin{align*}
    S_n&= \sum_{|I|=k} \sigma_IS_I , &
    \tilde{\s}{}_k&= \sum_{|I|=k} \h_{\sigma_IS_I} , &
    \s_n&= \sum_{|I|=k} \ad(\sigma_IS_I) .
  \end{align*}
  The operator $\alpha_1$ must satisfy the equation
  \begin{equation*}
    [d,\alpha_1 ] + \tilde{\s}{}_0 - \s_0 = 0 .
  \end{equation*}
  By Lemma \ref{homotopy}, we may choose $\alpha_1f=C(S_0,f)$.

  The operators $\alpha_n$, $n>1$, satisfy the recursion
  \begin{equation*}
    [d,\alpha_k ] + \sum_{i=0}^{k-2} \bigl( \tilde{\s}{}_i\alpha_{k-i-1} -
    \alpha_{k-i-1}\s_i \bigr) + \tilde{\s}{}_{k-1} - \s_{k-1} = 0 .
  \end{equation*}
  We make the Ansatz
  \begin{equation}
    \label{recurse}
    \alpha_k = \H^* \Bigl( \sum_{i=0}^{k-1} \bigl( \tilde{\s}{}_k\alpha_{k-i-1} -
    \alpha_{k-i-1}\s_i \bigr) + \tilde{\s}{}_{k-1} - \s_{k-1} \Bigr)
  \end{equation}
  The important point is that in this formula, $\H^*$ is only applied
  to elements of $F^1\FF$, where the algebraic Poincar\'e lemma yields
  a contracting homotopy of the differential $d$.

  To establish the theorem, we reformulate the recursion
  \eqref{recurse} as the solution of a fixed-point equation. We have
  \begin{equation*}
    \alpha = 1 + \alpha_1 - \H^* \bigr( \tilde{\s}A - A\s - \tilde{\s}{}_0 +
    \s_0 \bigr) .
  \end{equation*}
  It is easily seen that this reformulates the recursive definition of
  the operators $\alpha_k$.

  Let $\beta=[d,\alpha]+\tilde{\s}\alpha-\alpha\s$. Since $\H^*$ is being applied to
  elements of $F^1\FF$ in this expression, we may set $[d,\H^*]$ equal
  to $1$ in calculating $[d,\alpha]$. A straightforward calculation now
  shows that
  \begin{align*}
    [d,\alpha] &= [d,\alpha_1] - [ d,\H^* ] \bigr( \tilde{\s}\alpha - \alpha\s -
    \tilde{\s}{}_0 + \s_0 \bigr) \\
    &\quad + \H^* \bigr( [d,\tilde{\s}]\alpha - \tilde{\s}[d,\alpha] - [d,\alpha]\s - \alpha
    [d,\s] \bigr) \\
    &= \s_0 - \tilde{\s}{}_0 - \bigr( \tilde{\s}\alpha - \alpha\s -
    \tilde{\s}{}_0 + \s_0 \bigr) + \H^* \bigr( - \tilde{\s}^2\alpha -
    \tilde{\s}[d,\alpha] - [d,\alpha]\s + \alpha \s^2 \bigr) \\
    &= - \bigr( \tilde{\s}\alpha - \alpha\s \bigr) + \H^* \bigr( - \tilde{\s}^2\alpha -
    \tilde{\s}[d,\alpha] - [d,\alpha]\s + \alpha \s^2 \bigr) ,
  \end{align*}
  in other words,
  \begin{equation*}
    \beta = - \H^* \bigl( \tilde{\s}\beta - \beta\s \bigr) .
  \end{equation*}
  We see, by induction on $p$, that $\beta$ maps $\FF$ to $F^p\FF$. Since
  $F^p\FF$ vanishes if $p>n+1$, we conclude that $\beta=0$, proving the
  theorem.
\end{proof}

The resolution $(\FF,d+\tilde\s)$ of the complex $(\CF,\s)$ has the
disadvantage of not being a differential graded Lie superalgebra in
any natural way. It does present some technical advantages over the
sheaf of differential graded Lie superalgebras $(\FF,d+\tilde\s)$,
since it is the total complex of a bicomplex with differentials $d$
and $\tilde{\s}$, and there are many methods available to study the
cohomology of the evolutionary vector field $\tilde{\s}$.

\section{Covariant field theories}

Introduce formal variables $\{u^1,\dots,u^n\}$ of ghost number
$\gh(u^i)=2$, corresponding to the coefficients of a vector field
$u^i\p_i$ on $\R^n$. Define the densities
\begin{equation*}
  D_i = \< \p_i x^a , \xi_a \> , \quad i \in \{1,\dots,n\} ,
\end{equation*}
and let $D_u=u^iD_i$.

A \textbf{covariant} classical field theory is a solution
$\tint S_u\in\CF_\circ\[u\]$, of total ghost number $0$ and even parity, to
the covariant extension of the classical master equation
\begin{equation*}
  \half \( \tint S_u , \tint S_u \) = \tint D_u .
\end{equation*}
In particular, setting the variables $u^i$ to zero, we recover a
solution of the classical master equation.

\begin{lemma}
  \label{Di}
  The operator
  \begin{equation*}
    \iota_i = \sigma_i \Bigl( 1 - \sum_{\alpha\in\N^n} \xi_{a,\alpha} \p_{a,\alpha} \Bigr)
  \end{equation*}
  satisfies the equations $\iota_iD_j=0$ and
  \begin{equation*}
    d \iota_i + \iota_i d + \ad(D_i) = 0 .
  \end{equation*}
\end{lemma}
\begin{proof}
  The equation $\iota_iD_j=0$ is clear by inspection, while the formula
  for $d\iota_i+\iota_id$ is a special case of Lemma \ref{homotopy}, since
  \begin{equation*}
    \h_{D_i} = - \p_i
  \end{equation*}
  and
  \begin{equation*}
    \sum_{\alpha\in\N^n} \xi_{a,\alpha} \p_{a,\alpha} = C(D_i,-) .
    \qedhere
  \end{equation*}
\end{proof}

Define the covariant extension of the differential $d$ by the formula
$d_u = d + \iota_u$, where $\iota_u=u^i\iota_i$. The analogue of Proposition
\ref{lift} holds in this setting.
\begin{proposition}
  Given a covariant classical field theory, represented by a solution
  $\int S_u\in\CF_\circ\[u\]$ of the covariant classical master equation, there
  is a lift $S_u\in\FF_\circ\[u\]$ of $\int S_u$ to a solution of the
  \textbf{curved} Maurer-Cartan equation in $\FF_\circ\[u\]$:
  \begin{equation}
    \label{covariant}
    d_uS_u + \half \( S_u , S_u \) = D_u .
  \end{equation}
\end{proposition}
\begin{proof}
  We expand $S_u$ as the sum
  \begin{equation*}
    S_u = \sum_{k=0}^{n+1} S_k ,
  \end{equation*}
  where $S_k\in\FF_\circ^{-k,*}\[u\]$ has total degree $0$. Since
  $\half(\tint S_0,\tint S_0)-\tint D_u=0\in\CF_\circ\[u\]$, we may find an
  element $S_1\in\FF_\circ^{-1,*}\[u\]$ of total degree $0$ such that
  \begin{equation*}
    dS_1+\half\(S_0,S_0\)=D_u .
  \end{equation*}
  We now define $S_k\in\FF_\circ^{-k,*}\[u\]$ for $k>1$ by the recursive
  formula
  \begin{equation*}
    S_k = - \H \Bigl( \iota_uS_{k-1} + \half \sum_{j=0}^{k-1} \( S_j ,
    S_{k-j-1} \) \Bigr) .
  \end{equation*}
  It is vital to the proof that the argument of $\H$ lies in
  $F^1\FF_\circ$, where $[d,\H]=1$.

  Equivalently, $S$ is the fixed point of the equation
  \begin{equation*}
    S_u = S_0 + S_1 - \H \Bigl( \iota_uS_u+ \half \(S_u,S_u\) -
    \half \(S_0,S_0\) \Bigr) .
  \end{equation*}
  Since $dS_0=0$, we see that
  \begin{align*}
    dS_u &= dS_1 - \Bigl( \iota_uS_u + \half \(S_u,S_u\) - \half
    \(S_0,S_0\) \Bigr) + \H d \Bigl( \iota_uS_u+ \half \(S_u,S_u\) -
    \half \(S_0,S_0\) \Bigr) \\
    &= dS_1 + \half \(S_0,S_0\)  - \iota_uS_u - \half \(S_u,S_u\) + \H
    \Bigl( d \iota_uS_u - \(S_u,dS_u\) \Bigr) .
  \end{align*}
  Let $T_u=d_uS_u+\half\(S_u,S_u\)-D_u$. We find that
  \begin{equation*}
    T_u = - \H \Bigl( \iota_uT_u + \(S_u,T_u\) \Bigr) .
  \end{equation*}
  Here, we have used the Jacobi relation for the Soloviev bracket, in
  the form $\(S_u,\(S_u,S_u\)\)=0$. Since
  $\H:F^p\FF_\circ\to F^{p+1}\FF_\circ$ raises filtration degree in
  $\FF_\circ$, we see, by induction on $p$, that $T_u$ lies in
  $F^p\FF_\circ\[u\]$. Since $F^p\FF_\circ$ vanishes if $p>n+1$, we conclude that
  $T_u=0$, proving the theorem.
\end{proof}

\begin{remark}
  Just as the usual Maurer-Cartan equation is an abstracted form of
  the equation for a flat connection, the curved Maurer-Cartan
  equation may be thought of as an abstraction for the equation for a
  projectively flat connection.
\end{remark}

Introduce the operator $\s_u=\ad(S_u)$. It follows from Lemma \ref{Di}
that if $S_u$ satisfies the covariant classical master equation, then
$d_u+\s_u$ is a differential on $\FF\[u\]$. We now restrict attention
to the special case where $S_u$ is an affine function of the variables
$u^i$:
\begin{equation}
  \label{SG}
  S_u = S + u^i G_i .
\end{equation}
Note that $G_i\in\FF_\circ$ has ghost number $-2$ and even parity. The
equation \eqref{covariant} is equivalent to the classical master
equation \eqref{master} for $S$ and the two system of equations
\begin{equation*}
  dG_i + \iota_iS + \(S,G_i\) = D_i \, \quad 1\le i\le n ,
\end{equation*}
and
\begin{equation*}
  \iota_iG_j + \iota_jG_i + \(G_i,G_j\) = 0 , \quad 1\le i,j\le n .
\end{equation*}
Introduce the (degree $-1$, odd parity) operators
\begin{equation*}
  \Gamma_i = \iota_i + \ad(G_i)
\end{equation*}
on $\FF$. These pairwise graded commute (i.e.\ anticommute) with each
other and with the differential $d+\s$, and we have
\begin{equation*}
  d_u + \s_u = d + \s + u^i\Gamma_i .
\end{equation*}

The following theorem gives a purely algebraic transcription of the
construction of the Lagrangian density of AKSZ models in the framework
of covariant field theories. Let $\Gamma=\Gamma_1\dots\Gamma_n$.
\begin{theorem}
  Let $\rho\in\FF^{0,n}_\circ$ be a cocycle for the differential $\s$, such
  that $\(\rho,\Gamma\rho\)=0$. Then
  \begin{equation*}
    S^\rho_u = S_u + \Gamma\rho
  \end{equation*}
  is a solution of the covariant classical master equation
  \eqref{covariant}.
\end{theorem}
\begin{proof}
  By the equations $\s\rho=(d+\s)\rho=0$ and $[d+\s,\Gamma_i]=0$, it follows that
  \begin{equation*}
    (d_u+\s_u)\Gamma\rho = u^i\Gamma_i(\Gamma\rho) .
  \end{equation*}
  By the equation $\Gamma^i\Gamma^j+\Gamma^j\Gamma^i=0$, we see that
  $\Gamma_i(\Gamma\rho)=\Gamma_i(\Gamma_1\dots_n\rho)=0$, and hence that
  \begin{equation*}
    (d_u+\s_u)\Gamma\rho = 0 .
  \end{equation*}
  
  Since $\Gamma_i$ is a graded derivation of the Soloviev bracket and
  $\Gamma_i(\Gamma\rho)=0$, we see that
  \begin{equation*}
    \Gamma \( \rho , \Gamma\rho \) = \( \Gamma\rho , \Gamma\rho \) ,
  \end{equation*}
  showing that
  \begin{equation*}
    \( \Gamma\rho , \Gamma\rho \) = 0 ,
  \end{equation*}
  and completing the proof.
\end{proof}

Since the induced action of $\Gamma_i$ on $\CF$ equals
$\ad(\int G_i)$, we obtain the following corollary.

\begin{corollary}
  Let $\int\rho\in\CF^n_\circ$ be a cocycle for the differential $\s$, of ghost
  number and parity $n$, such that
  $\bigl(\bigl(\int G_1,\dots,\bigl(\int G_n,\int \rho \bigr)\dots\bigr),
  \int\rho\bigr)=0$. Then
  \begin{equation*}
    \tint S^\rho = \tint S + \bigl( \tint G_1,\dots,\bigl(\tint
    G_n,\tint \rho \bigr)\dots\bigr)
    \in \CF_\circ
  \end{equation*}
  is a solution of the classical master equation \eqref{Master}.
\end{corollary}

The remainder of this section is not used elsewhere in the article,
but serves to relate the above formulas to the approach of
\cite{covariant}, which handled the special case where $n=1$. In loc.\
cit., we expressed the twist of a solution of the covariant classical
master equation by a cochain $\rho$ as a gauge transformation.

If $W\in\FF_\circ[u]$ has total ghost number $-1$, it generates a gauge
transformation of solutions of the covariant classical master
equation:
\begin{equation*}
  S_u \bullet W = S_u + \sum_{k=0}^\infty \frac{1}{(k+1)!} (-\ad(W))^k (d_u+\s_u) W .
\end{equation*}
If $W$ satisfies the auxilliary equation $\( W , (d_u+\s_u)W \) = 0$,
then $S_u\bullet\rho=S_u+(d_u+\s_u)W$.

More generally, this formula still applies if $W\in\FF_\circ(u)$ is a
rational function of the variables $u^i$ such that
$(d_u+\s_u)W\in\FF_\circ[u]$ is a \textbf{polynomial} function of the
variables $u^i$ and $\(W,(d_u+\s_u)W\)=0$. This way of constructing
solutions of the Maurer-Cartan equation is a nonlinear analogue of the
idea that a coboundary in a complex is a cocycle in any subcomplex
that contains it. We may express the solution $S_u+ \Gamma\rho$ of the
covariant classical master equation as a singular gauge transformation
of $S_u$.

Make the Ans\"atz
\begin{equation*}
  W = \frac{1}{n!(u\cdot u)} \epsilon_{i_1}{}^{i_2\dots i_n} u^{i_1}
    \Gamma_{i_2} \dots \Gamma_{i_n} \rho ,
\end{equation*}
where $u\cdot u=\delta_{ij}u^iu^j$. It follows that
$(d_u+\s_u)W = \Gamma\rho$, and that
\begin{align*}
  \(W,(d_u+\s_u)W\) &= \frac{1}{n!(u\cdot u)} \epsilon_{i_1}{}^{i_2\dots i_n}
  u^{i_1} \( \Gamma_{i_2} \dots \Gamma_{i_n} \rho , \Gamma\rho \) \\
  &= \frac{1}{n!(u\cdot u)} \epsilon_{i_1}{}^{i_2\dots i_n} u^{i_1} \Gamma_{i_2} \dots
  \Gamma_{i_n} \( \rho , \Gamma\rho \) = 0 .
\end{align*}
We see that $S_u^\rho = S_u \bullet W$ is obtained from $S_u$ by a singular
gauge transformation of Maurer-Cartan elements in the curved Lie
algebra $(\FF_\circ[u],d_u+\s_u)$.

\section{The covariant abelian Chern-Simons action}

In the remainder of this paper, we consider the special case of
classical Chern-Simons theory. This covariant field theory is quite
special: it is a first-order field theory in which every physical
field is bosonic. In this section, we solve the covariant classical
master equation for abelian Chern-Simons theory, and calculate its
Batalin-Vilkovisky cohomology. We build upon these results in the next
section, where we solve the covariant classical master equation for
non-abelian classical Chern-Simons theory.

The lift of the action for the abelian Chern-Simons field theory to
$\FF_\circ$ is
\begin{equation*}
  \Pi = \half \epsilon^{ijk} \< A_i , \p_jA_k \> + \half \< \p_i c , A^{+i}
  \> - \half \< c , \p_i A^{+i} \> \in \FF_\circ^{0,0}
\end{equation*}
It is clear that $d\Pi=0$, and that $\iota_i\Pi=0$ for $1\le i\le n$. The
classical master equation follows on verification that
$\half\(\Pi,\Pi\)=0$, which is established by a straightforward
calculation. Let $\pi=\ad(\Pi)$, and $\tilde\pi=\h_\Pi$.

Define evolutionary vector fields
\begin{equation*}
  \f^i = \frac12 \biggl( - \Bigl\< c , \frac{\p~}{\p A_i} \Bigr\> +
  \epsilon^{ijk} \Bigl\< A_j , \frac{\p~}{\p A^{+k}} \Bigr\> + \Bigl\< A^{+i}
  , \frac{\p }{\p c^+} \Bigr\> \biggr) , \quad 1\le i\le 3 .
\end{equation*}
(These vector fields are not Hamiltonian.) It is easily seen that
$[\pi,\f^i]=[\f^i,\f^j]=0$ for $1\le i,j\le 3$. The following result shows
that for the abelian Chern-Simons theory, there is a particularly
simple choice for the automorphism $\alpha$ of $\FF$ introduced in
Theorem~\ref{main}.
\begin{proposition}
  \label{alpha}
  $\exp\bigl( \sigma_i \f^i \bigr) (d+\tilde\pi) = (d+\pi) \exp\bigl( \sigma_i\f^i
  \bigr)$
\end{proposition}
\begin{proof}
  Observe that $C(\Pi,-) = - \sigma_i\f^i$. It follows that
  \begin{align*}
    \pi &= \tilde\pi - [d,C(\Pi,-)] \\
    &= \tilde\pi + \p_i\f^i .
  \end{align*}
  The proof is completed by observing that
  \begin{equation*}
    [[d,\sigma_i\f^i],\sigma_j\f^j] = [\pi,\sigma_i\f^i] = 0 .
    \qedhere
  \end{equation*}
\end{proof}

The Batalin-Vilkovisky cohomology $H^*(\CF,\pi)$ of the abelian
Chern-Simons theory is easily calculated. We must calculate the
cohomology of the complex $(\FF^*,d+\pi)$, which by
Proposition~\ref{alpha} is isomorphic to the cohomology of the double
complex $(\FF^{*,*},d,\tilde\pi)$. We take the cohomology of $\tilde\pi$
first. The vector field $\tilde\pi$ is given by the formula
\begin{equation*}
  \tilde\pi = \pr \biggl( - \Bigl\< \p_ic , \frac{\p~}{\p A_i} \Bigr\>
  + \epsilon^{ijk} \Bigl\< \p_iA_j , \frac{\p~}{\p A^{+k}} \Bigr\>
  + \Bigl\< \p_i A^{+i} , \frac{\p~}{\p c^+} \Bigr\> \biggr) .
\end{equation*}
Following the proof of Barnich et al.\ \cite{BBH}*{}, we filter the
jet coordinates according to the following table.

\begin{equation*}
  \begin{tabu}{|c|l|} \hline
    \mathmakebox[0.2in][c]{0} & \quad \mathmakebox[0.6in][l]{\p^\alpha c^+} \\
    & \quad \mathmakebox[0.6in][l]{\p^\alpha A^{+1}} \alpha_1>0 \\ \hline
    1 & \quad \mathmakebox[0.6in][l]{\p^\alpha A^{+1}} \alpha_1=0 \\
    & \quad \mathmakebox[0.6in][l]{\p^\alpha A^{+2}} \\
    & \quad \mathmakebox[0.6in][l]{\p^\alpha A^{+3}} \\
    & \quad \mathmakebox[0.6in][l]{\p^\alpha A_2} \alpha_1>0 \\
    & \quad \mathmakebox[0.6in][l]{\p^\alpha A_3} \alpha_1+\alpha_2>0 \\ \hline
    2 & \quad \mathmakebox[0.6in][l]{\p^\alpha A_1} \\
    & \quad \mathmakebox[0.6in][l]{\p^\alpha A_2} \alpha_1 = 0 \\
    & \quad \mathmakebox[0.6in][l]{\p^\alpha A_3} \alpha_1+\alpha_2=0 \\
    & \quad \mathmakebox[0.6in][l]{\p^\alpha c}
    \mathmakebox[1.1in][l]{\alpha_1+\alpha_2+\alpha_3>0} \\ \hline
    3 & \quad \mathmakebox[0.6in][l]{c} \\ \hline
  \end{tabu}
\end{equation*}
The leading order term of the differential $\tilde\pi$ in this
filtration is
\small
\begin{multline*}
  \sum_\alpha \biggl( \Bigl\< \p^{\alpha+1} A^{+1} , \frac{\p~}{\p(\p^\alpha c^+)}
  \Bigr\> + \Bigl\< \p^{\alpha+1}A_2 , \frac{\p~}{\p(\p^\alpha A^{+3})} \Bigr\>
  - \Bigl\< \p^{\alpha+1}A_3 , \frac{\p~}{\p(\p^\alpha A^{+2})} \Bigr\> -
  \Bigl\< \p^{\alpha+1}c , \frac{\p~}{\p(\p^\alpha A_1)} \Bigr\>  \biggr) \\
  + \sum_{\{\alpha \mid \alpha_1=0\}} \biggl( \Bigl\< \p^{\alpha+2}A_3 , \frac{\p~}{\p(\p^\alpha
    A^{+1})} \Bigr\> - \Bigl\< \p^{\alpha+2}c , \frac{\p~}{\p(\p^\alpha A_2)}
  \Bigr\> \biggr) - \sum_{\{\alpha\mid \alpha_1+\alpha_2=0\}} \Bigl\< \p^{\alpha+3}c ,
  \frac{\p~}{\p(\p^\alpha A_3)} \Bigr\> ,
\end{multline*}
\normalsize Taking cohomology, we obtain the graded vector space
$C^*(\g)$ of cochains on the Lie algebra $\g$, generated by $c$. The
differential on $C^*(\g)$ vanishes, since $\g$ is abelian. Denote by
$\tau_{\ge1}C^*(\g)$ the associated reduced complex, obtained by
quotienting $C^*(\g)$ by $C^0(\g)$.

The complex $(H^*(\FF,\tilde\pi),d)$ may be identified with the
augmented de Rham complex
\begin{equation*}
  0 \to \R\eta \to \Om^0 \o C^*(\g) \xrightarrow{d} \Om^1 \o C^*(\g)
  \xrightarrow{d} \to \Om^2 \o C^*(\g) \xrightarrow{d} \Om^3 \o C^*(\g) \longrightarrow
  0 .
\end{equation*}
Applying the Poincar\'e lemma, we see that the $E_2$-page
$H^p(H^q(\FF,\tilde\pi),d)$ for the spectral sequence of the double
complex equals
\begin{equation*}
  H^p(H^q(\FF,\tilde\pi),d) \cong
  \begin{cases}
    C^q(\g) , & p=-3,q>0 , \\
    0 , & \text{otherwise.}
  \end{cases}
\end{equation*}
Clearly, the spectral sequence is convergent (since $E^{pq}_0$
vanishes for $p$ outside the bounded interval $[-4,0]$), and collapses
at the $E_2$-page. This proves the following theorem.
\begin{theorem}
  If $\g$ is abelian, the Batalin-Vilkovisky cohomology of the
  classical Chern-Simons theory associated to $\g$ equals
  \begin{equation*}
    H^k(\CF,\pi) \cong \tau_{\ge1}C^{k+3}(\g) .
  \end{equation*}
\end{theorem}

We now prove that the abelian Chern-Simons theory is a covariant field
theory; this will allow us to obtain explicit formulas for cocycles in
$\CF$ representing elements of the Batalin-Vilkovisky cohomology. The
Chern-Simons theory treats the fields and antifields of the theory
equally, as components of different degrees of a superconnection. It
turns out that the formulas for the contravariant field theory are
simplified if we use the following modification of $D_i$ which treats
the fields and antifields equally:
\begin{equation*}
  D_i = \half \< \p_i A^j , A^+_j \> - \half \< A^j , \p_i A^+_j \>  +
  \half \< \p_i c , c^+ \> - \half \< c , \p_i c^+ \>  , \quad i \in
  \{1,\dots,n\} .
\end{equation*}
The associated operator $\iota_i$ is given by the formula
\begin{align*}
  \iota_i &= \sigma_i \bigl( 1 - C(D_i,-) \bigr) \\
  &= - \half \sigma_i ( \E - 2 ) ,
\end{align*}
where $\E$ is the Euler vector field, which in this setting equals
\begin{equation*}
  \E = \pr \biggl( \Bigl\< A_i , \frac{\p~}{\p A_i} \Bigr\> + \Bigl\<
  A^{+i} , \frac{\p~}{\p A^{+i}} \Bigr\> + \Bigl\< c , \frac{\p~}{\p
    c} \Bigr\> + \Bigl\< c^+ , \frac{\p~}{\p c^+} \Bigr\> \biggr) .
\end{equation*}

Let
\begin{equation*}
  G_i = \half \epsilon_{ijk} \< A^{+j} , A^{+k} \> + \< A_i , c^+ \> \in
  \FF_\circ^{0,-2} .
\end{equation*}
It is clear that $dG_i=0$ for $1\le i\le n$, and that $\iota_iG_j=0$ for
$1\le i,j\le n$.

\begin{lemma}
  $\( \Pi , G_i \) = D_i$
\end{lemma}
\begin{proof}
  In the proof, we make use of the formula
  $\delta^k_r \epsilon^{pqr} \epsilon_{ijk} = \delta^p_i \delta^q_j - \delta^p_j \delta^q_i$. We have
  \begin{align*}
    \( \epsilon^{pqr} \< A_p , \p_qA_r \> , G_i \) &=
    \( \epsilon^{pqr} \< A_p , \p_qA_r \> , \half \epsilon_{ijk} \< A^{+j}
    A^{+k} \> \) \\
    &= \bigl( \delta^q_i \delta^r_k - \delta^q_k \delta^r_i \bigr) \< \p_qA_r , A^{+k} \>
    + \bigl( \delta^p_i \delta^q_k - \delta^p_k \delta^q_i \bigr) \< A_p , \p_q A^{+k} \>  \\
    &= \< \p_iA_j , A^{+j} \> - \< \p_jA_i , A^{+j} \> + \< A_i , \p_j
    A^{+j} \> - \< A_j , \p_i A^{+j} \> \\
    \( \< \p_p c , A^{+p} \> , G_i \) &=
    \( \< \p_p c , A^{+p} \> , \< A_i , c^+ \> \) \\
    &= \< \p_i c , c^+ \> + \< \p_jA_i , A^{+j} \> , \\
    \( \< c , \p_p A^{+p} \> , G_i \) &=
    \( \< c , \p_p A^{+p} \> , \< A_i , c^+ \> \) \\
    &= \< c , \p_i c^+ \> + \< A_i , \p_j A^{+j} \> .
  \end{align*}
  The lemma follows.
\end{proof}

\begin{corollary}
  $dG_i + \( \Pi , G_i \) + \iota_i \Pi = D_i$
\end{corollary}

The Soloviev bracket $\(G_i,G_j\)$ is the symmetrization of the
Soloviev bracket
\begin{equation*}
  \( \half \epsilon_{ik\ell} \< A^{+k} , A^{+\ell} \> , \< A_j , c^+ \> \)
  = - \epsilon_{ijk} \< A^{+k} , c^+ \>
\end{equation*}
in the indices $i$ and $j$, hence it vanishes. This completes the
proof of the following result.
\begin{proposition}
  The abelian Chern-Simons theory $\Pi+u^iG_i$ is a covariant field
  theory.
\end{proposition}

Let $\Gamma=\Gamma_1\Gamma_2\Gamma_3$. The operation
$\Gamma:\tau_{\ge1}C^{*+3}(\g)\to(\FF^*,d+\pi)$ induces a quasi-isomorphism of
complexes. Similarly, the operation
\begin{equation*}
  f(c) \in C^k(\g) \mapsto \tint \Gamma f(c) = \tint \bigl( \tint G_1 , \bigl(
  \tint G_2 , \bigl( \tint G_3 , \tint f(c) \bigr)\bigr)\bigr) \in
  \CF^{k-3}
\end{equation*}
induces a quasi-isomorphism of complexes from $\tau_{\ge1}C^{*+3}(\g)$ to
$(\CF^*,\pi)$. Introduce the operations
\begin{align*}
  \CD(A_i) &= \pr \Bigl\< A_i , \frac{\p~}{\p c} \Bigr\> + \half \sigma_i
  (\E-2) , &
  \CD(A^{+i}) &= \pr \Bigl\< A^{+i} , \frac{\p~}{\p c} \Bigr\> , &
  \CD(c^+) &= \pr \Bigl\< c^+ , \frac{\p~}{\p c} \Bigr\>
\end{align*}
on $\FF$. For $f(c)\in C^k(\g)\subset\FF^{0,k}$, we have the formula
\begin{align}
  \label{GGG}
  \Gamma f(c) &= - \CD(A_1)\CD(A_2)\CD(A_3)f(c) +
  \CD(A_i)\CD(A^{+i})f(c) - \CD(c^+)f(c) \\
  & \in \FF^{k-3} . \notag
\end{align}

Given a pair $f(c),g(c)\in C^*(\g)$ of Lie algebra cochains, the bracket
$\( \Gamma f(c), g(c) \)$ is given by the formula
\begin{equation}
  \label{Poisson}
  \{f,g\} = \Bigl\< \frac{\p f(c)}{\p c} , \frac{\p g(c)}{\p c} \Bigr\> .
\end{equation}
In the next section, we extend the results of this section to when
$\g$ is non-abelian.

\section{The covariant non-abelian Chern-Simons action}

In this section, we construct the covariant extension of the classical
field theory of non-abelian Chern-Simons theory, following the approach
of Section~3.

Suppose that the inner product space $\g$ considered in the last
section underlies a real Lie algebra with bracket $[-,-]$ and
invariant inner product $\<-,-\>$. (We do not assume that the inner
product is definite, only that it is nondegenerate.) We consider the
local functional
\begin{equation*}
  \rho = \tfrac{1}{6} \< c , [c,c] \> .
\end{equation*}
This is certainly a cocycle for the abelian Chern-Simons theory of the
previous section, of ghost number $3$.

By \eqref{GGG}, we see that
\begin{align*}
  \Gamma\rho &= \tfrac{1}{6} \epsilon^{ijk} \< A_i , [A_j,A_k] \> - \half
  \< c^+ , [c,c] \> - \< A^{+i} , [A_i,c] \> \\
  &+ \tfrac{1}{4} \sigma_i \Bigl( \< A^{+i} , [c,c] \> +
  \epsilon^{ijk} \< [A_j , A_k ] , c \> \Bigr) \\
  &- \tfrac{1}{16} \sigma_i\sigma_j \epsilon^{ijk} \< A_k , [c,c] \> - \tfrac{1}{288}
  \sigma_i \sigma_j \sigma_k \epsilon^{ijk} \< c , [c,c] \> .
\end{align*}
By \eqref{Poisson}, we see that
\begin{equation*}
  \(\rho,\Gamma\rho\) = - \tfrac{1}{4} \< [c,c] , [c,c] \> = -
  \tfrac{1}{4} \< c, [ c, [c,c]] \> = 0 ,
\end{equation*}
since $[c,[c,c]]$ vanishes by the Jacobi relation. In this way, we
obtain the following explicit solution of the covariant classical
master equation for Chern-Simons theory:
\begin{align*}
  S_u &= \Pi + \Gamma\rho + u^iG_i \\
  &= \cs_3(\AA) - \tfrac{1}{4} \sigma_k \Bigl( \< A^{+k} , [c,c] \> +
  \epsilon^{ijk} \< c , [A_i , A_j ] \> \Bigr) \\
  &\quad - \tfrac{1}{16} \sigma_j\sigma_k \epsilon^{ijk} \< A_i , [c,c] \> + \tfrac{1}{288}
  \sigma_i \sigma_j \sigma_k \epsilon^{ijk} \< c , [c,c] \> + u^iG_i .
\end{align*}

We conclude that the morphism
$\Gamma:\tau_{\ge1}C^{*+k}(\g) \to \FF^*$ of \eqref{GGG} works as well in the
non-abelian case, giving rise to a quasi-isomorphism of complexes of
degree $-3$ from the reduced Lie algebra cochains
$\bigl( \tau_{\ge1}C^*(\g),\delta \bigr)$ of the Lie algebra $\g$ to the
Batalin-Vilkovisky complex $(\FF,d+\s)$ of classical Chern-Simons
theory:
\begin{equation*}
  H^k(\CF,\s) \cong
  \begin{cases}
    H^{k+3}(\g) , & k\ge-2 , \\
    0 , & k<-2 .
  \end{cases}
\end{equation*}
If $\g$ is semisimple, Whitehead's lemma shows that
$H^1(\g) \cong H^2(\g)\cong0$, and hence that $H^*(\CF,\s)$ vanishes in
negative degree. These results were obtained Barnich et al.\
\cite{BBH}*{Section 14}.

Barnich and Grigoriev \cite{BG} have given a geometric interpretation
of this result. (Their results more generally for any AKSZ field
theory: Chern-Simons theory was the original example that led to the
introduction of the class of AKSZ field theories.) If $\g$ is a
reductive Lie algebra over $\R$ with invariant inner product
$\<-,-\>$, the graded manifold $\g[1]$ has the graded commutative
algebra of Lie algebra cochains $C^*(\g)$ as its ring of
functions. Denote the coordinate functions by $c\in C^1(\g)\o\g$. The
inner product determines a symplectic form on $\g[1]$,
\begin{equation*}
  \Om = \half \< dc , dc \> ,
\end{equation*}
with associated Poisson bracket
\begin{equation*}
  \{ f , g \} = (-1)^{|f|} \Bigl\< \frac{\p f}{\p c} , \frac{\p g}{\p
    c} \Bigr\> .
\end{equation*}
The differential $\delta$ on the complex of Lie algebra cochains may be
interpreted geometrically as the Hamiltonian vector field associated
to the cochain $\rho$. In this setting, Barnich and Grigoriev prove the
existence of a quasi-isomorphism
\begin{equation*}
  \bigl( \tau_{\ge1}C^{*+3}(\g) , \delta \bigr) \to (\CF,\s) .
\end{equation*}
In this section, we have obtained a lift of this quasi-isomorphism of
complexes to a quasi-isomorphism of shifted differential graded Lie
algebras.

In the appendix, we prove the purely algebraic result that the
differential graded $(-2)$-shifted Lie algebra $C^*(\g)$ associated to
a reductive Lie algebra $\g$ with invariant inner product $\<-,-\>$ is
homotopy abelian. This implies the following result for the sheaf of
1-shifted differential graded Lie algebras $\CF$ on $\R^3$ associated
to the corresponding Chern-Simons theory.
\begin{theorem}
  The Batalin--Vilkovisky complex $(\CF,\s)$ associated to the
  Chern-Simons theory of a semisimple Lie algebra is homotopy abelian,
  and there is a quasi-isomorphism of differential graded 1-shifted
  Lie algebras $(\tau_{\ge1}H^*(\g),0)\hookrightarrow(\CF,\s)$.
\end{theorem}

It follows that the deformation theory of $(\CF,\s)$ is unobstructed:
the moduli space of deformations is naturally isomorphic to the space
of invariant inner products on $\g$. Note that the first obstruction
is identically zero, since $H^4(\g)=0$ for a semisimple Lie algebra
$\g$.

\appendix

\section{The Chevalley-Eilenberg complex of a semisimple Lie algebra}

If $\g$ is a reductive Lie algebra, Chevalley and Eilenberg prove that
the inclusion
\begin{equation*}
  C^*(\g)^\g \subset C^*(\g)
\end{equation*}
of invariant cochains is a quasi-isomorphism. Since the differential
vanishes on $C^*(\g)^\g$, there is a natural isomorphism
$H^*(\g)\cong C^*(\g)^\g$.

An invariant inner product $\<-,-\>$ on $\g$ induces a $(-2)$-shifted
Poisson bracket on $C^*(\g)$, making $C^*(\g)$ into a differential
graded $(-2)$-shifted Lie algebra. Let $\theta\in C^1(\g,\g)$ be the
1-cochain on $\g$ given by the identity map; think of $\theta$ as a
coordinate system on the graded manifold $\g[1]$, and cochains as
functions on $\g[1]$. The Poisson bracket of $f\in C^k(\g)$ and
$C^\ell(\g)$ equals
\begin{equation*}
  \{ f , g \} = (-1)^k \Bigl\< \frac{\p f}{\p\theta} , \frac{\p g}{\p\theta}
  \Bigr\> .
\end{equation*}
The differential $\delta$ of $C^*(\g)$ equals $\ad(\rho)$, where is
$\rho$ is the invariant 3-cochain
$\tfrac16 \< \theta , [\theta,\theta] \>$. In particular,
$\delta\theta=-\half[\theta,\theta]$. Let $\Theta=\half[\theta,\theta]$.

The purpose of this appendix is to record a proof of the following
result.
\begin{proposition}
  \label{appendix}
  If $\g$ is semisimple and $f$ and $g$ are invariant cochains,
  $\{f,g\}=0$.
\end{proposition}
This result implies that the differential graded Lie algebra $C^*(\g)$
is formal, and that the differential graded Lie algebra $H^*(\g)$ is
abelian. It follows that the same statement holds for the sheaf of
differential graded Lie algebras $(\CF,\s)$ for the Chern-Simons
theory with semisimple Lie algebra $\g$.

Let $S(\g)$ be the space of polynomials on the finite-dimensional Lie
algebra $\g$ with invariant inner product $\<-,-\>$. The gradient
$\nabla P\in S(\g)\o\g$ of a polynomial $P(x)$ is given by the formula
\begin{equation*}
  \< \nabla P(x) , y \> = \frac{d}{ds}\Bigm|_{s=0} P(x+sy) .
\end{equation*}
Let $I(\g)$ be the subspace of invariant polynomials: $P\in I(\g)$ is
invariant if and only if for all $x,y\in\g$,
\begin{equation*}
  \< \nabla P(x) , [y,x] \> = 0 .
\end{equation*}
Grade the algebra $I(\g)$ by placing polynomials homogeneous of degree
$\ell$ in degree $2\ell$, making it into a graded commutative algebra, and
let $I_+(\g)$ be the subspace of invariant polynomials with vanishing
constant term.

Define a map
\begin{equation*}
  \tau : I(\g) \to C^\ast(\g)[-1]
\end{equation*}
by the formula
\begin{equation*}
  \tau P = (2\ell-1)^{-1} \< \nabla P(\Theta) , \theta \> .
\end{equation*}
As an example, conside the quadratic form $Q(x)=\half\<x,x\>$
associated to the invariant bilinear form $\<x,y\>$ on $\g$, gives
rise to the cocycle $\tau Q = \rho$.
\begin{lemma}
  \label{I}
  Let $\g$ be a reductive Lie algebra. Let $P\in I(\g)$. Then
  \begin{enumerate}[1),nosep]
  \item $P(\Theta)=0$, and
  \item $\p(\tau P)/\p\theta=\nabla P(\Theta)$.
  \end{enumerate}
\end{lemma}
\begin{proof}
  Since $\g$ is reductive, $I(\g)$ is spanned by trace polynomials of
  the form
  \begin{equation*}
    P_{V,\ell}(x) = \Tr_V(\rho(x)^\ell) ,
  \end{equation*}
  where $\rho:\g\to\gl(V)$ is a finite-dimensional representation of
  $\g$ and $\ell>0$. (See Humphreys \cite{Humphreys}*{Section~23.1}; this
  is a theorem of Chevalley, which follows from the Weyl character
  formula.) But $\rho(\Theta)=\half[\rho(\theta),\rho(\theta)]=\rho(\theta)^2$, so
  \begin{equation*}
    P_{V,\ell}(\Theta) = \Tr_V(\rho(\theta)^{2\ell}) = \half \Tr_V [ \rho(\theta) , \rho(\theta)^{2\ell-1} ]
    = 0 .
  \end{equation*}
  It follows that
  \begin{equation*}
    \tau P_{V,\ell} = \frac{\ell}{2\ell-1} \Tr_V( \rho(\Theta)^{\ell-1}\rho(\theta) ) =
    \frac{\ell}{2\ell-1} \Tr_V(\rho(\theta)^{2\ell-1}) ,
  \end{equation*}
  and
  \begin{equation*}
    \frac{\p(\tau P_{V,\ell})}{\p\theta} = \ell \rho(\Theta)^{\ell-1} .
    \qedhere
  \end{equation*}
\end{proof}

\begin{corollary}
  If $\g$ is reductive, $\tau$ maps $I^{2\ell}(\g)$ to the space of cocycles
  $Z^{2\ell-1}(\g)\subset C^{2\ell-1}(\g)$, and the restriction of
  $\tau$ to $I^+(\g)^2$ vanishes.
\end{corollary}
\begin{proof}
  Let $P$ be a homogeneous invariant polynomial of degree $\ell>0$. Since
  $d\theta=\Theta$, we have
  \begin{equation*}
    (2\ell-1) \, d(\tau P) = \< \nabla P(\Theta) , \Theta \> = \ell \, P(\Theta) ,
  \end{equation*}
  which vanishes by the lemma.

  Consider $P,Q\in I_+(\g)$, of degree $\ell$ and $m$ respectively. The map
  $T$ vanishes on $PQ$:
  \begin{equation*}
    (2\ell+2m-1) \, \tau(PQ) = \< \nabla P(\Theta) , \theta \> \, Q(\Theta) + P(\Theta) \, \< \nabla Q(\Theta)
    , \theta \> = 0 .
    \qedhere
  \end{equation*}
\end{proof}

We may choose a sequence $P_i\in I^{2\ell_i}(\g)$, $1\le i\le r$, of
homogeneous invariant polynomials whose images form a basis of
$I_+(\g)/I_+(\g)^2$. The numbers $(\ell_1,\dots,\ell_r)$ are called the
exponents of $\g$, and the natural number $r$ is the \textbf{rank} of
$\g$. Chevalley proves that $I(\g)$ is a polynomial algebra with
generators $P_i$.  The map $\tau:I(\g)\to C^*(\g)[-1]$ induces a morphism
from $I_+(\g)/I_+(\g)^2$ to $Z^*(\g)[-1]$. Borel proves that $H^*(\g)$
is an exterior algebra generated by $[\tau P_i]\in H^{2\ell_i-1}(\g)$.

For the classical groups, the generators $P_i\in I(\g)$ are trace
polynomials associated to the fundamental representation, with a
single exception, the Pfaffian $\Pf(x)\in I(\so(2m))$, which is the
difference of two trace polynomials associated to the two half-spinor
representations $\rho_\pm:\so(2m)\to S_\pm$:
\begin{equation*}
  \Pf(x) = i^n \bigl( \Tr_{S_+}\bigl( \rho_+(x)^m \bigr) -
  \Tr_{S_-}\bigl( \rho_-(x)^m \bigr) \bigr) .
\end{equation*}
Chevalley's result is only really needed in the proof of Lemma~\ref{I}
to handle the possible presence of the exceptional Lie algebras.

It follows that the inclusion of the subalgebra of $C^*(\g)$ spanned
by the cocycles
\begin{equation*}
  \tau P_{i_1} \dots \tau P_{i_k} \in Z^{2\ell_{i_1}+\cdots+2\ell_{i_k}-k}(\g) , \qquad 1\le
  i_1<\dots<i_k\le r ,
\end{equation*}
induces a quasi-isomorphism of dg commutative algebras
$H^*(\g)\hookrightarrow C^*(\g)$. The image of this morphism is the subalgebra
$C^*(\g)^\g$ of invariant cochains.

Proposition \ref{appendix} is now a consequence of the following
lemma.
\begin{lemma}
  If $P\in I(\g)$ and $Q\in I(\g)$ are homogeneous of degree $\ell,m>1$,
  \begin{equation*}
    \{ \tau P , \tau Q \} = 0 .
  \end{equation*}
\end{lemma}
\begin{proof}
  By Lemma~\ref{I}
  \begin{equation*}
    \{ \tau P , \tau Q \} = \Bigl\< \frac{\p(\tau P)}{\p\theta} , \frac{\p(\tau
      Q)}{\p\theta} \Bigr\> = \< \nabla P (\Theta) , \nabla Q(\Theta) \> .
  \end{equation*}
  Since $\< \nabla P , \nabla Q \>$ is invariant and homogeneous of degree
  $\ell+m-2>0$, Lemma~\ref{I} implies the vanishing of
  $\< \nabla P(\Theta) , \nabla Q(\Theta) \>$.
\end{proof}


\begin{bibdiv}

\begin{biblist} 

\bib{AS}{article}{
   author={Axelrod, Scott},
   author={Singer, I. M.},
   title={Chern-Simons perturbation theory},
   conference={
      title={Proceedings of the XXth International Conference on
      Differential Geometric Methods in Theoretical Physics, Vol.\ 1, 2},
      address={New York},
      date={1991},
   },
   book={
      publisher={World Sci. Publ., River Edge, NJ},
   },
   date={1992},
   pages={3--45},
}

\bib{BBH}{article}{
   author={Barnich, Glenn},
   author={Brandt, Friedemann},
   author={Henneaux, Marc},
   title={Local BRST cohomology in gauge theories},
   journal={Phys. Rep.},
   volume={338},
   date={2000},
   number={5},
   pages={439--569},
   issn={0370-1573},
}

\bib{BG}{article}{
   author={Barnich, Glenn},
   author={Grigoriev, Maxim},
   title={A Poincar\'{e} lemma for sigma models of AKSZ type},
   journal={J. Geom. Phys.},
   volume={61},
   date={2011},
   number={3},
   pages={663--674},
   issn={0393-0440},
}

\bib{CF}{article}{
   author={Cattaneo, Alberto S.},
   author={Felder, Giovanni},
   title={On the AKSZ formulation of the Poisson sigma model},
   note={EuroConf\'{e}rence Mosh\'{e} Flato 2000, Part II (Dijon)},
   journal={Lett. Math. Phys.},
   volume={56},
   date={2001},
   number={2},
   pages={163--179},
   issn={0377-9017},
}

\bib{CMR}{article}{
   author={Cattaneo, Alberto S.},
   author={Mnev, Pavel},
   author={Reshetikhin, Nicolai},
   title={Classical BV theories on manifolds with boundary},
   journal={Comm. Math. Phys.},
   volume={332},
   date={2014},
   number={2},
   pages={535--603},
   issn={0010-3616},
}

\bib{Darboux}{article}{
   author={Getzler, Ezra},
   title={A Darboux theorem for Hamiltonian operators in the formal calculus
   of variations},
   journal={Duke Math. J.},
   volume={111},
   date={2002},
   number={3},
   pages={535--560},
}

\bib{covariant}{article}{
   author={Getzler, Ezra},
   title={Covariance in the Batalin-Vilkovisky formalism and the
   Maurer-Cartan equation for curved Lie algebras},
   journal={Lett. Math. Phys.},
   volume={109},
   date={2019},
   number={1},
   pages={187--224},
   issn={0377-9017},
}

\bib{GM}{article}{
   author={Goldman, William M.},
   author={Millson, John J.},
   title={The deformation theory of representations of fundamental groups of
   compact K\"{a}hler manifolds},
   journal={Inst. Hautes \'{E}tudes Sci. Publ. Math.},
   number={67},
   date={1988},
   pages={43--96},
   issn={0073-8301},
}

\bib{Humphreys}{book}{
   author={Humphreys, James E.},
   title={Introduction to Lie algebras and representation theory},
   series={Graduate Texts in Mathematics, Vol. 9},
   publisher={Springer-Verlag, New York-Berlin},
   date={1972},
   pages={xii+169},
}

\bib{Olver}{book}{
   author={Olver, Peter J.},
   title={Applications of Lie groups to differential equations},
   series={Graduate Texts in Mathematics},
   volume={107},
   publisher={Springer-Verlag, New York},
   date={1986},
   pages={xxvi+497},
   isbn={0-387-96250-6},
}

\bib{Soloviev}{article}{
   author={Soloviev, Vladimir O.},
   title={Boundary values as Hamiltonian variables. I. New Poisson brackets},
   journal={J. Math. Phys.},
   volume={34},
   date={1993},
   number={12},
   pages={5747--5769},
   issn={0022-2488},
}

\bib{Sullivan}{article}{
   author={Sullivan, Dennis},
   title={Infinitesimal computations in topology},
   journal={Inst. Hautes \'{E}tudes Sci. Publ. Math.},
   number={47},
   date={1977},
   pages={269--331 (1978)},
   issn={0073-8301},
}

\end{biblist}
\end{bibdiv}
\end{document}